\documentclass[runningheads]{llncs}
\usepackage{graphicx}
\spnewtheorem*{conjecture*}{Conjecture}{\itshape}{\rmfamily}

\begin{document}
\title{On Conflict-Free Replicated Data Types and Equivocation in Byzantine Environments
\thanks{The work of Hannes Hartenstein was supported by the Helmholtz Association. The work of Florian Jacob and Saskia Bayreuther was partially supported by the Carl Zeiss Foundation.}
}
\titlerunning{On CRDTs and Equivocation in Byzantine Environments}
%
\author{Florian Jacob\orcidID{0000-0002-5739-8852} \and
Saskia Bayreuther\orcidID{0000-0001-5197-0507} \and
Hannes Hartenstein\orcidID{0000-0003-3441-3180}}
%
\authorrunning{F. Jacob et al.}
\institute{Karlsruhe Institute of Technology (KIT), Germany
\\
\email{\{florian.jacob,saskia.bayreuther,hannes.hartenstein\}@kit.edu}}
%
\maketitle              
\begin{abstract}
We explore the property of \textit{equivocation tolerance} for Conflict-Free Replicated Data Types (CRDTs).
We show that a subclass of CRDTs is equivocation-tolerant and can thereby cope with any number of Byzantine faults:
Without equivocation detection, prevention or remediation, they still fulfill strong eventual consistency (SEC).
We also conjecture that there is only one operation-based CRDT design supporting non-commutative operations that fulfills SEC in Byzantine environments with any number of faults.

\end{abstract}
\section{Introduction}

As the name suggests, Conflict-free Replicated Data Types (CRDTs) provide powerful properties: in particular, updates can be applied without further coordination of replicas, and recovery from network partitions can be done with ease.
While crash-fault environments are typically assumed for CRDTs, it is natural to investigate whether or under which conditions these desired CRDT properties also hold in Byzantine environments.

Recent works on CRDTs in Byzantine environments have followed different paths.
The work stretches from classical assumptions of an honest two-thirds majority~\cite{zhao2016optimistic} introducing coordination (e.g., Byzantine-tolerant causal-order broadcast~\cite{auvolat2021byzantine}) to coordination-free, sybil-resistant CRDTs using broadcast based on the happened-before relation as directed, acyclic graphs~\cite{kleppmann2020bec,jacob2021analysis}.

In this brief announcement, we relate the notion of equivocation to CRDTs and show under which conditions a subclass of CRDTs is equivocation-tolerant in Byzantine environments.
We show that, due to equivocation tolerance, all state-based and certain operation-based CRDTs tolerate any number of Byzantine faults while depending on rather mild assumptions on the communication layer.

\emph{Equivocation} is the act of a Byzantine replica sending different updates, which appear valid on their own, to different recipients where it should have sent the same update~\cite{chun2007attested}.
While equivocation can only be detected globally or with two equivocated updates, a \emph{valid update} is an update that is protocol-conforming when viewed on its own.
\emph{Omission} occurs if a crash- or Byzantine-faulty replica sends an update to only a strict subset of all protocol-intended recipients.

A system provides \emph{Non-Equivocation} if it prevents Byzantine replicas from performing Equivocation.
We say that an algorithm is \emph{equivocation-tolerant} if it neither needs to detect, prevent, nor remedy equivocation to ensure its provided guarantees beyond what is needed to cope with omission.
It follows that correct replicas need to be able to locally detect an invalid update.

\section{CRDTs and Equivocation Tolerance}

CRDTs provide \emph{Strong Eventual Consistency (SEC)}, consisting of \emph{Strong Convergence} (“correct replicas that have delivered the same updates have an equivalent state”), \emph{Eventual Delivery} (“an update delivered at some correct replica is eventually delivered to all correct replicas”), and \emph{Termination} (“All method executions terminate”)~\cite{shapiro2011conflict}.
Equivocation mainly threatens Strong Convergence: either the notion of which updates are the same, or the application order of updates can be equivocated.
Intuitively, \emph{equivocation tolerance} for CRDTs means that validity of updates is defined in a way that \emph{any two valid updates are conflict-free}, which allows applying \emph{any} valid update directly, without threatening Strong Convergence.
As uncoordinated update application is then sufficient, threats to Termination like coordinated equivocation detection and remediation can be avoided.
We argue that CRDTs are equivocation-tolerant if they remain conflict free when facing any number of byzantine faults.
With equivocation tolerance, only omission faults remain, threatening Eventual Delivery.

Below, we show that all state-based CRDTs and a subclass of operation-based CRDTs are equivocation-tolerant.
Consequently, an arbitrary number of Byzantine replicas cannot harm SEC.
We assume authenticated channels between a static group of replicas participating in the CRDT, and that all correct replicas form a connected component in the communication graph.

\subsection{State-based CRDTs and Equivocation Tolerance}

\begin{theorem}
\label{th:state-based-equivocation-tolerance}
All state-based CRDTs provide equivocation tolerance.
\end{theorem}

\begin{proof}
State-based CRDTs are based on defining a join-semilattice made of all valid states.
They only send their current state without metadata, which means that an update is valid if and only if it is part of the semilattice, which is locally verifiable.
Due to the commutativity of the join relation and the partial order of the semilattice, any two valid updates cannot conflict with each other, as both can be merged in an arbitrary order with the same result.
An equivocation consisting of $d$ differing updates can thereby be treated as $d$ independent updates for which omission has occurred.
\qed
\end{proof}

\begin{lemma}
\label{lm:state-based-Byzantine-tolerance}
State-based CRDTs ensure Strong Eventual Consistency for all correct replicas in an environment with $n$ replicas of which an arbitrary number $f$ exhibit Byzantine faults, i.e., have fault tolerance $n > f$.
\end{lemma}

\begin{proof}
Due to Theorem~\ref{th:state-based-equivocation-tolerance}, we can treat equivocation as omission, and are left to ensure Eventual Delivery.
State-based CRDTs gossip their current state regularly and unsolicitedly to all other replicas.
As their current state indirectly contains all updates they have received and merged before, updates not sent directly to a specific replica will eventually reach that replica indirectly via correct replicas.
\qed
\end{proof}

\subsection{Operation-based CRDTs and Equivocation Tolerance}

While state-based CRDTs need few assumptions to work in Byzantine environments with any number of faults, only a subset of operation-based CRDTs can do so, and yet they need a number of additional assumptions.
Specifically, for equivocation tolerance, all potentially conflicting operations have to be locally detected as invalid.
To avoid inconsistencies due to equivocation, we require that valid operations have an \emph{inherent identity}:
Any two operations that are the same as defined by the protocol need to lead to the same outcome when applied.
Inherent identity can be ensured by deriving the identity of an operation from its content, either by comparing the full content, or by verifiable unique identifiers gained via content addressing:
the identifier of an operation is the hash of its content.
Thereby, identical operations are not applied twice when received twice.

Operation-based CRDTs require that non-commutative operations are applied in causal order.
To avoid inconsistencies due to equivocation on operation ordering, we also require that valid operations have an \emph{inherent ordering}:
Either all operations are commutative and no ordering is needed, or the datatype semantics inherently records a happened-before relation, i.e., the potential causal order in the datatype payload.
If done via content addressing, Byzantine attackers cannot tamper with the happened-before relation, as hashes verifiably prove that an operation referenced by another operation via its hash has happened before the other operation.

\begin{theorem}
\label{th:operation-based-equivocation-tolerance}
Operation-based CRDTs that provide inherent identity and inherent ordering of operations are equivocation-tolerant.
\end{theorem}

\begin{proof}
Without the above requirements, operations of operation-based CRDTs contain metadata that can introduce conflicts, e.g., a non-unique operation identifier or wrong happened-before relation.
Inherent identity and inherent ordering prevent those conflicts and thereby ensure that equivocations can be treated as omissions.
However, inherent identity and ordering must be locally verifiable, which is the case for all explained mechanisms.
An equivocation that breaks SEC and leads to a permanent inconsistent state must either break inherent identity or inherent order, breaking an assumption of their mechanisms: if ensured by hashing, the equivocation can be reduced to a hash collision.
\qed
\end{proof}

Without periodic gossiping of state-based CRDTs, operation-based CRDTs need to be able to handle omissions to ensure Eventual Delivery.
\emph{Omission Handling} can rely on a payload-recorded happened-before relation to detect missing elements in the causal order.
Using content addressing, those missing elements can be re-requested verifiably from other replicas.
Alternatively, CRDTs can periodically gossip the set of all received operations.
The gossiping approach can be formalized and made more efficient through the happened-before relation and hash chaining~\cite{kleppmann2020bec}.
We note that this approach essentially uses a state-based set CRDT to synchronize all operations, benefiting from the Byzantine tolerance of all state-based CRDTs shown in Lemma~\ref{lm:state-based-Byzantine-tolerance}.

\begin{lemma}
\label{lm:operation-based-Byzantine-tolerance}
Omission-handling, equivocation-tolerant operation-based CRDTs ensure Strong Eventual Consistency for all correct replicas in an environment with $n$ replicas of which an arbitrary number $f$ exhibit Byzantine faults, i.e., have fault tolerance $n > f$.
\end{lemma}

\begin{proof}
Due to Theorem~\ref{th:operation-based-equivocation-tolerance}, we can treat equivocation as omission.
Using one of the explained omission handling mechanisms, equivocation-tolerant operation-based CRDTs can ensure that omitted operations are eventually delivered.
\qed
\end{proof}

\section{Discussion and Conclusion}

We showed that some CRDTs can be used in Byzantine environments with an arbitrary number of faults by leveraging their equivocation tolerance.
We now discuss how these CRDTs can be brought into practice.

In classical CRDT use cases, all parties have full permissions to perform any CRDT update, including the deletion of current states.
Thus, deleting is not a Byzantine act.
Therefore, when dealing with CRDTs in Byzantine environments, the focus is typically on grow-only CRDTs.

In environments with byzantine majorities, a CRDT with non-commutative operations has to record the happened-before relation in the payload to ensure the causal order independently of the broadcast order.
The only way known to us to ensure the causal order in a locally verifiable, Byzantine-tolerant way is to use hash chaining.
The happened-before relation being a partial order inherently leads to a directed, acyclic graph of all operations.
To efficiently ensure Eventual Delivery, it is natural to employ the happened-before relationship recorded in the graph by re-requesting missing parent operations and only gossiping childless operations.
This line of thought leads us to the following conjecture:

\begin{conjecture*}
A hash-chained directed acyclic graph as described in~\cite{kleppmann2020bec,jacob2021analysis} is the only operation-based CRDT with non-commutative operations that provides SEC
for any number of Byzantine faults, i.e., has fault tolerance $n > f$.
\end{conjecture*}

State-based CRDTs are easy to deploy in Byzantine environments, as shown in Lemma~\ref{lm:state-based-Byzantine-tolerance}, based on their unconditional equivocation tolerance.
However, the identity of updates gets lost since only states are propagated in the system.
It is not possible to reconstruct the update that led to a new state, which makes it impossible to prove which replica performed which CRDT updates and whether it was allowed to do so.
Hence, access control on the different access and update methods of state-based CRDTs, giving different permissions to different participating replicas, is impossible to enforce.
This makes state-based CRDTs suitable for decentralized systems like the Newsgroup system where any user can write new articles and reply to old ones.
The current Newsgroup system is susceptible to equivocation, as users are trusted to not assign the same identifier to different articles.
If based on an equivocation-tolerant CRDT, SEC could be guaranteed.

In contrast to state-based CRDTs, as shown in Lemma~\ref{lm:operation-based-Byzantine-tolerance}, operation-based CRDTs require additional properties to provide SEC in Byzantine environments.
However, with operation-based CRDTs, the original caller of an update method can be determined and verified with authentication mechanisms like digital signatures.
Therefore, operation-based CRDTs provide the necessary prerequisites for access control, which makes them easier to deploy in systems like instant messaging.
In summary, state-based grow-only CRDTs can play out their strengths in public, permissionless Byzantine systems, while operations-based CRDTs have more restrictions, but allow for permissions with finer granularity of replicas.

In other fields, equivocation is handled not by tolerance but by prevention.
In Byzantine fault-tolerant agreement protocols, \emph{non-equivocation} can be fulfilled by preventing Byzantine processes from creating two valid messages with the same identifier when identifiers are created using monotonic counters that are located inside trusted hardware components \cite{clement2012limited}.
In distributed ledger technologies like blockchains, equivocation usually means creating a fork.
While it is possible to mine two blocks that have the same hash value, the probability that the resulting branches co-exist is very low due to probabilistic leader election mechanisms like proof of work.

In this brief announcement, we analyzed the reasons why and under which circumstances a subclass of CRDTs can be moved from the crash fault model to a Byzantine fault model using the notion of equivocation.
We showed that even when faced with an arbitrary number of Byzantine faults, this subclass can keep the characteristic traits of CRDTs, like efficiency, low coordination effort, and Strong Eventual Consistency.


%
%
\bibliographystyle{splncs04}
\bibliography{paper}
\end{document}